\documentclass[a4paper,UKenglish,cleveref, autoref]{lipics-v2019}

\theoremstyle{plain}
\newtheorem{conjecture}[theorem]{Conjecture}

\newcommand{\final}[1]{}
\newcommand{\arxiv}[1]{#1}

\title{Hardness of Bichromatic Closest Pair with Jaccard Similarity} 

\titlerunning{Hardness of Bichromatic Closest Pair with Jaccard Similarity}

\author{Rasmus Pagh}{BARC and IT University of Copenhagen, Denmark}{pagh@itu.dk}{https://orcid.org/0000-0002-1516-9306
}{This research has received funding from the European Research Council under the European Union’s 7th Framework Programme (FP7/2007-2013) / ERC grant agreement no.~614331.}

\author{Nina Mesing Stausholm}{BARC and IT University of Copenhagen, Denmark}{nimn@itu.dk}{https://orcid.org/0000-0002-4322-7163}{}

\author{Mikkel Thorup}{BARC and University of Copenhagen, Denmark}{mikkel2thorup@gmail.com}{https://orcid.org/0000-0001-5237-1709}{}

\funding{This work was supported by Investigator Grant 16582, Basic Algorithms Research Copenhagen
(BARC), from the VILLUM Foundation.}

\authorrunning{R. Pagh, N.\,M. Stausholm  and M. Thorup}

\Copyright{Rasmus Pagh, Nina M. Stausholm and Mikkel Thorup}



\category{}

\relatedversion{}

\supplement{}

\acknowledgements{We want to thank A. Rubinstein for helping us understand the background of his results in \cite{rubinstein2018hardness}.}

\nolinenumbers 

\EventEditors{Michael A. Bender, Ola Svensson, and Grzegorz Herman}
\EventNoEds{3}
\EventLongTitle{27th Annual European Symposium on Algorithms (ESA 2019)}
\EventShortTitle{ESA 2019}
\EventAcronym{ESA}
\EventYear{2019}
\EventDate{September 9--11, 2019}
\EventLocation{Munich/Garching, Germany}
\EventLogo{}
\SeriesVolume{144}
\ArticleNo{80}

\begin{document}

\maketitle

\begin{abstract}
Consider collections $\mathcal{A}$ and $\mathcal{B}$ of red and blue sets,
respectively. Bichromatic Closest Pair is the problem of finding a
pair from $\mathcal{A}\times \mathcal{B}$ that has similarity higher than a given threshold
according to some similarity measure. Our focus here is the classic Jaccard similarity
$|\textbf{a}\cap \textbf{b}|/|\textbf{a}\cup \textbf{b}|$ for $(\textbf{a},\textbf{b})\in \mathcal{A}\times \mathcal{B}$.

We consider the approximate version of the problem where we are given thresholds $j_1>j_2$ and wish to return
a pair from $\mathcal{A}\times \mathcal{B}$ that has Jaccard similarity higher than
$j_2$ if there exists a pair in $\mathcal{A}\times \mathcal{B}$ with Jaccard similarity at
least $j_1$. The classic locality sensitive hashing (LSH) algorithm of Indyk and Motwani (STOC '98), instantiated with the MinHash LSH function of Broder et al., solves this problem in $\tilde O(n^{2-\delta})$ time if $j_1\ge j_2^{1-\delta}$. In particular, for $\delta=\Omega(1)$, the approximation ratio $j_1/j_2=1/j_2^{\delta}$ increases polynomially in $1/j_2$.

In this paper we give a corresponding hardness result. Assuming the
Orthogonal Vectors Conjecture (OVC), we show that there cannot be a
general solution that solves the Bichromatic Closest Pair problem in $O(n^{2-\Omega(1)})$ time
for $j_1/j_2=1/j_2^{o(1)}$. Specifically, assuming OVC, we prove that for
any $\delta>0$ there exists an $\varepsilon>0$ such that Bichromatic
Closest Pair with Jaccard similarity requires time
$\Omega(n^{2-\delta})$ for any choice of thresholds
$j_2<j_1<1-\delta$, that satisfy $j_1\le j_2^{1-\varepsilon}$.
\end{abstract}

\section{Introduction}
Twitter is a well-known social network, in which a user can connect to other users by \emph{following} them \cite{goel2013discovering}. Users can read and write messages called \emph{tweets} of up to 280 characters.
An important service that Twitter provides is helping users discover other users that they might like to follow, by making suggestions. This service is called the \emph{You might also want to follow}-service and is better known as the WTF (Who To Follow) recommender system \cite{gupta2013wtf}. In order to suggest connections that the user might like, they should be similar to the user's existing connections. As an example, if a user is already connected to Cristiano Ronaldo, Twitter might suggest Lionel Messi as a new connection, since the connection to Ronaldo hints that the user likes famous soccer players.
Hence, we need a way to decide if a connection is similar to an existing connection. We might for instance suggest a new connection if the tweets are similar to the tweets of an existing connection or if the connection has a lot of the same followers as an existing connection.

\medskip

The main challenge is to find similar connections when the number of user accounts increases drastically and the task is particularly difficult when the similarity does not need to be significant, i.e., when we look for connections that have only little in common with existing ones, while they may still be of interest to the particular user \cite{goel2013discovering}. This leads us to the notion of \emph{similarity search}, which concerns the general problem of searching for similar objects in a collections of objects. Often we consider these objects as sets representing some concept or entity. An object could for example be a document that is represented by a set of words. Hence, we talk about \emph{set similarity search}.

\medskip

There are several versions of the problem addressing different situations. In this paper we consider a batched version of set similarity search, namely the Bichromatic Closest Pair which can be informally described as follows:

Suppose we are given collections $\mathcal{A}$ and $\mathcal{B}$, each of $n$ sets from a universe of size $O(\log n)$. We refer to the sets in $\mathcal{A}$ as \emph{red} and the sets in $\mathcal{B}$ as \emph{blue}.
Bichromatic Closest Pair is the problem of finding the pair consisting of a red and a blue set that is closest with respect to some distance or similarity measure. We will concern ourselves with Jaccard similarity, which is defined for a pair of sets $(\textbf{a},\textbf{b})\in \mathcal{A}\times \mathcal{B}$ as
\begin{align}
\label{def:jaccardsim}
J(\textbf{a},\textbf{b})=\frac{\vert \textbf{a}\cap \textbf{b}\vert}{\vert \textbf{a}\cup \textbf{b}\vert}=\frac{\vert \textbf{a}\cap \textbf{b}\vert}{\vert \textbf{a}\vert+\vert \textbf{b}\vert-\vert \textbf{a}\cap \textbf{b}\vert}.
\end{align}
In particular, we consider the following \emph{decision version} of Bichromatic Closest Pair with Jaccard similarity: decide whether there exists a pair $(\mathbf{a},\mathbf{b})\in \mathcal{A}\times \mathcal{B}$ such that $J(\mathbf{a},\mathbf{b})\ge j_1$ or if all pairs $(\mathbf{a},\mathbf{b})\in \mathcal{A}\times \mathcal{B}$, has $J(\mathbf{a},\mathbf{b})< j_2$ for given thresholds $j_1$ and $j_2$.

\medskip
It is well-known that we can solve Bichromatic Closest Pair with Jaccard similarity for thresholds satisfying $j_1\ge j_2^{1-\delta}$ in time $O(n^{2-\delta})$ (see Section~\ref{sec:relatedwork}). 
In particular, for $\delta=\Omega(1)$, the
approximation ratio $j_1/j_2=1/j_2^{\delta}$ increases polynomially in $1/j_2$. 
In this paper, we will present a corresponding hardness
result. The hardness is conditioned on one of the most well-known
and widely believed hypotheses, namely the Orthogonal Vectors
Conjecture \cite{Williams18}.

\begin{conjecture}(Orthogonal Vectors Conjecture (OVC))
For every $\delta>0$ there exists $c=c(\delta)$ such that given two collections $\mathcal{A},\mathcal{B}\subset\{0,1\}^m$ of cardinality $n$, where $m=c\log n$, deciding if there is a pair $(\mathbf{a},\mathbf{b})\in \mathcal{A}\times \mathcal{B}$ such that $\mathbf{a}\cdot \mathbf{b}=0$ requires time $\Omega(n^{2-\delta})$.
\end{conjecture}

Assuming OVC, we show that there cannot be a general
solution that solves the Bichromatic Closest Pair problem with Jaccard
similarity in $O(n^{2-\Omega(1)})$ time for
$j_1/j_2=1/j_2^{o(1)}$. More specifically, we show

\begin{theorem}
\label{thm:mainthm}
Assuming the Orthogonal Vectors Conjecture (OVC), the following holds: for any $\delta>0$, there exists an $\varepsilon>0$ such that for any given $j_2<j_1<1-\delta$ satisfying $j_1\le j_2^{1-\varepsilon}$, solving Bichromatic Closest Pair with Jaccard similarity for $n$ red and $n$ blue sets for sets from a universe of size
$\ln(n) / j_2^{O(\log 1/j_1)}$
for thresholds $j_1$ and~$j_2$ requires time  $\Omega(n^{2-\delta})$.
\end{theorem}

The dependence of $\varepsilon$ on $\delta$ is unspecified because the function $c(\delta)$ in OVC is not specified, see discussion in \arxiv{Appendix~\ref{app:computations}.}\final{Appendix~B in the full version on ArXiv \cite[App.~B]{arxiveversion}.}

\subsection{Techniques and Related Work}
\label{sec:relatedwork}
Similarity search can be performed in several ways -- a popular technique is Locality Sensitive Hashing (LSH) \cite{indyk1998approximate} which attempts to collect similar items in buckets in order to reduce the number of sets needed to check similarity against.
We can for example use Broder's MinHash \cite{broder1997resemblance} with locality sensitive hashing to solve Bichromatic Closest Pair with Jaccard similarity in time $\tilde{O}(n^{2-\varepsilon})$ when $j_1\ge j_2^{1-\varepsilon}$ for any $\varepsilon$. This is done by ensuring that the collision probability for pairs with similarity $j_2$ is $1/n$ and the collision probability for pairs with similarity $j_1$ is $1/n^{1-\varepsilon}$. Hashing $n^{1-\varepsilon}$ times means that we find a pair with similarity $j_1$ if one exists.
The ChosenPath method presented in \cite{christiani2017set} also uses the LSH framework to solve Bichromatic Closest Pair with Braun-Blanquet similarity in time $\tilde{O}(n^{2-\varepsilon})$ for thresholds $j_1\ge j_2^{1-\varepsilon}$.

\medskip
The proof of Theorem \ref{thm:mainthm} will be based on a result by Rubinstein~\cite{rubinstein2018hardness}: Assuming the Orthogonal Vectors Conjecture, a $(1+\varepsilon)$-approximation to Bichromatic Closest Pair with Hamming, Edit or Euclidean distance requires time $\Omega(n^{2-\delta})$. 
The required approximation factor $1+\varepsilon$ depends on~$\delta$, and tends to 1 as $\delta$ tends to zero. 
We translate this into an equivalent conditional lower bound for Jaccard similarity for certain constants
$j_1$ and $j_2$.

In order to handle smaller subconstant values of $j_1$ and $j_2$ we use a technique that we call squaring, which allows us to increase the gap in similarities between pairs with high Jaccard similarity and pairs with low Jaccard similarity by computing the cartesian product of a binary vector with itself. A similar technique is used in \cite{valiant2015finding} by Valiant. His technique is called \emph{tensoring} and is used to amplify the gap between small and large inner products of vectors. 
We also see a similar technique in the LSH framework with MinHash, where we use concatenation of hash values (which are sampled set elements) to amplify the difference in collision probability, and hence in the Jaccard similarity. 

\medskip
 Combining two simple reductions with the above squaring we show that for any $\delta$, we can always find $\varepsilon$ such that Bichromatic Closest Pair with Jaccard similarity cannot be solved in time $O(n^{2-\delta})$ for any pair $j_1,j_2<1-\delta$ when $j_1\le j_2^{1-\varepsilon}$.
Contrast this with the above LSH upper bound of $\tilde{O}\left(n^{2-\delta}\right)$ for $j_1\ge j_2^{1-\delta}$. We also know that there are parts of the parameter space where $j_1=j_2^{1-\delta}$ that can be solved in $\tilde{O}\left(n^{2-\delta-\Omega(1)}\right)$ time, see the discussion in \cite{christiani2017set}.
While LSH with MinHash is not the fastest possible algorithm in terms of the exponent achieved,  it has been unclear how far from optimal it might be.

\paragraph*{Other related work}
Very recently, Chen and Williams \cite{chen2018equivalence} showed that assuming the OVC we cannot additively approximate our Bichromatic Closest Pair problem with Jaccard similarity.
It might be possible to use Chen and Williams as a base for showing our main theorem, but this would require reductions quite different from the ones presented in this paper.

An earlier of result of Chen \cite{chen2018hardness} shows that it is not possible (under OVC) to compute a $\left(d/\log n\right)^{o(1)}$-approximation to Maximum Inner Product (Max-IP) with two sets of $n$ vectors from $\{0,1\}^d$ in time~$O(n^{2-\Omega(1)})$.

\section{Preliminaries}

\subsection{Notation}
\label{sec:notation}
We will occasionally consider a set, $\mathbf{x}$, from a finite universe
$U=\{u_1,...,u_{\vert U\vert}\}$ as a vector $\mathbf{v}$ of dimension $\vert U\vert$ such that $v_i=[u_i\in x]$, in Iverson notation. We call this vector the \emph{characteristic vector for $\mathbf{x}$}. Hence, we refer to the set of indexes and the universe interchangeably. We denote the Hamming weight of a binary vector $\mathbf{v}$ by $\vert \mathbf{v}\vert$.
In the following, we will not only index vectors with integers, but also with vectors of integers. Hence, we will consider vectors of dimension $d^2$ with entries $v_{ij}$, for $i=(i_1,...,i_d)$ and $j=(j_1,...,j_d)$.
 
\subsection{Bichromatic Closest Pair}
Recall Jaccard similarity as is defined in (\ref{def:jaccardsim}). 
We define Bichromatic Closest Pair with Jaccard similarity for thresholds $t_1$ and $t_2$ as follows: Let $U$ be a universe of size $O(\log n)$. Given collections $\mathcal{A}$ and $\mathcal{B}$, each of $n$ sets from $U$, and thresholds $t_2<t_1<1$, we will consider the problem of finding a pair of sets $(\mathbf{a},\mathbf{b})\in \mathcal{A}\times \mathcal{B}$ with $J(\mathbf{a},\mathbf{b})\ge t_2$ if there exists a pair $(\mathbf{a}^*,\mathbf{b}^*)\in \mathcal{A}\times \mathcal{B}$ with $J(\mathbf{a}^*,\mathbf{b}^*)\ge t_1$. If all pairs have $J(\mathbf{a},\mathbf{b})< t_2$, we must not return any pair of sets.

\subsection{Useful instances of Bichromatic Closest Pair}
\label{sec:baseinstances}
The following lemma corresponds to Theorem 4.1 in \cite{rubinstein2018hardness} and will form the basis of our results.
It includes the important properties of the instances constructed in the proof the theorem, which we will use actively to prove our own Theorem \ref{thm:mainthm}.

\begin{lemma}
\label{lem:rubinstances}
Assume OVC. Given $\delta>0$, there exist $\varepsilon>0$ and values $h_1,h_2$ where $h_2=(1+\varepsilon)h_1$ such that Bichromatic Closest Pair with Hamming distance for thresholds $h_1$ and $h_2$ requires time $\Omega(n^{2-\delta})$ for instances with $n$ red and $n$ blue sets from a universe of size $O(\log n)$. There are instances that require this time with the following properties, where we let $T=O\left(\frac{1}{\varepsilon}\right)$ and $m=O(\log n)$:
\begin{itemize}
    \item All red sets have size $Tm$ and all blue sets have size $m$.
    \item The thresholds $h_1$ and $h_2$ are $m(T-1)$ and $mT$, respectively.
    \item All sets in the instance come from a universe of size $2Tm$.
\end{itemize}
\end{lemma}

In particular, the lemma states that we cannot compute a $(1+\varepsilon)$-approximation to Bichromatic Closest Pair with Hamming distance in truly subquadratic time. We will extend this result in a few steps, using the properties of the hard instances, to achieve Theorem \ref{thm:mainthm}. 
\subsection{Hardness of Bichromatic Closest Pair with Jaccard similarity}
In order to prove Theorem \ref{thm:mainthm}, we need the following lemma, which extends Lemma \ref{lem:rubinstances} in the natural way to Jaccard similarity.
\begin{lemma}
\label{lem:basecase}
Assuming OVC, we have the following: For any $\delta>0$ there exist $j_1,j_2$ with $j_1=2\cdot j_2$ such that Bichromatic Closest Pair with Jaccard similarity with thresholds $j_1$ and $j_2$ requires time $\Omega(n^{2-\delta})$.
\end{lemma}
\begin{proof}
We use instances as described in Lemma \ref{lem:rubinstances}. First, note that
\[
J(\mathbf{a},\mathbf{b})=\frac{\vert \mathbf{a}\cap \mathbf{b}\vert}{\vert \mathbf{a}\cup \mathbf{b}\vert}=\frac{\frac{\vert \mathbf{a}\vert +\vert \mathbf{b}\vert-d_H(\mathbf{a},\mathbf{b})}{2}}{\vert \mathbf{a}\vert+\vert \mathbf{b}\vert-\frac{\vert \mathbf{a}\vert +\vert \mathbf{b}\vert-d_H(\mathbf{a},\mathbf{b})}{2}}=\frac{\vert \mathbf{a}\vert +\vert \mathbf{b}\vert-d_H(\mathbf{a},\mathbf{b})}{\vert \mathbf{a}\vert +\vert \mathbf{b}\vert+d_H(\mathbf{a},\mathbf{b})}
\]
which implies that letting
\begin{align*}
    j_1&=\frac{Tm +m-m(T-1)}{Tm +m+m(T-1)}=\frac{1}{T}\qquad \text{and}\qquad
    j_2=\frac{Tm +m-Tm}{Tm +m+Tm}=\frac{1}{2T+1},
\end{align*}
we cannot solve Bichromatic Closest Pair with Jaccard similarity in time $O(n^{2-\delta})$.
Since $T=O\left(\frac{1}{\varepsilon}\right)$, as mentioned in Lemma \ref{lem:rubinstances}, we get a lower bound for the approximation factor:
\[
\frac{\frac{1}{T}}{\frac{1}{2T+1}}=\frac{2T+1}{T}=2+\frac{1}{T}=2+\Omega(\varepsilon).
\]
In particular, we achieve hardness of a 2-approximation. 
\end{proof}

\section{Overview of reductions used}
\label{sec:reductions}
We prove Theorem \ref{thm:mainthm} by combining several reductions into one. So let $(\mathcal{A},\mathcal{B})$ be any instance of Bichromatic Closest Pair with Jaccard similarity as described in Lemma \ref{lem:rubinstances}.
We give a brief introduction to each of these reductions -- note that all reductions are self-reductions. We give the details of the proof and the use of each reduction in Section~\ref{sec:result}. \arxiv{Further details can be found in Appendix \ref{app:computations}.}\final{Further details can be found in Appendix B in the full version on ArXiv \cite[App.~B]{arxiveversion}.}

\begin{itemize}
    \item \textbf{Adding common elements to sets:} Adding common elements to all sets in collections $\mathcal{A}$ and $\mathcal{B}$ increases the Jaccard similarity between any pair of red and blue sets.
    \item \textbf{Adding different elements to sets:} Adding elements to all sets in $\mathcal{A}$ decreases the Jaccard similarity between any pair of red and blue sets.
    \item \textbf{Squaring:} Consider all sets by their characteristic vector. We define squaring as follows: given vector $\mathbf{a}=(a_1,...,a_d)$
    the squared vector has entries
    \[
    a'_{ij}=a_i\cdot a_j\qquad \text{for $i,j\in\{1,...,d\}$}.
    \]
    The resulting vector $\mathbf{a}'$, which is the characteristic vector for $\mathbf{a}\times \mathbf{a}$, has dimension $d^2$ as described in Section~\ref{sec:notation}. Vector $\mathbf{a}'$ can equivalently be considered as a set from a universe of size $d^2$. We will use this reduction iteratively to reduce the Jaccard similarity between any pair of vectors in the instance of Bichromatic Closest Pair.
    \item \textbf{Sampling:} We will use sampling to reduce the size of the universe after each step of squaring. Hence, we consider squaring and sampling as a single reduction which first squares the vectors and then samples from the resulting vectors. We will use the squaring-and-sampling reduction iteratively.
\end{itemize}

\section{The squaring-and-sampling reduction -- details}
\label{sec:squaringdetails}
In the proof of Theorem \ref{thm:mainthm} we will take any instance of Bichromatic Closest Pair with Jaccard similarity with the properties described in Lemma \ref{lem:rubinstances} and use the squaring reduction described in Section~\ref{sec:reductions} to decrease the Jaccard similarity of every pair of sets in the instance. We will argue that a solution for the new instance also provides a solution for the original instance.
When squaring all sets, the Jaccard similarity between any pair of sets will decrease, so we need to capture this change in the thresholds, such that a solution for the new instance implies a solution for the initial instance. When squaring the sets in $\mathcal{A}$ and $\mathcal{B}$, the size of the sets will be squared and it is easy to see that so will the size of the intersection. Hence, the Jaccard similarity of a pair $(\mathbf{a},\mathbf{b})$ after squaring $i$ times, $\left(\mathbf{a}_{i},\mathbf{b}_{i}\right)$ is 
\begin{align}\label{JaccardSquaredNoSampling}
J\left(\mathbf{a}_i,\mathbf{b}_i\right)=\frac{\vert \mathbf{a}\cap \mathbf{b}\vert^{2^i}}{\vert \mathbf{a}\vert^{2^i}+\vert \mathbf{b}\vert^{2^i}-\vert \mathbf{a}\cap \mathbf{b}\vert^{2^i}}.
\end{align}
In order to keep down the size of the universe, we need to sample after each step of squaring. This might incur a small error in the Jaccard similarity. The next few sections will bound this error.
From this point, we will denote the squaring-and-sampling reduction by $f$. Hence, applying the reduction $f$ to a set, $\mathbf{v}$, $i$ times will yield a set $\mathbf{f(v,i)}$.

\subsection{Subsampling}
We bound the error incurred in each of $\vert \mathbf{a}\cap \mathbf{b}\vert$, $\vert \mathbf{a}\vert$ and $\vert \mathbf{b}\vert$ and combine these with a union bound to get a bound on the error in the Jaccard similarity. We shall see that when sampling sufficiently many elements from the universe the sets are taken from, we get that with high probability a solution for the constructed instance will provide a valid solution for the original instance.

\medskip
The following lemmas will help us show that sampling after squaring will not distort the similarity of the resulting vectors too much. 

\begin{lemma}
\label{lem:vectorsizeNEW}
Let $0<m'<m<1$ and let $\mathbf{p}$ be a set from a universe of size $s^2$ for an integer $s$. Assume that $\left(m'\cdot s\right)^2\le \vert \mathbf{p}\vert\le \left(m\cdot s\right)^2$. Sample $s'$ elements from the universe uniformly at random, $\mathbf{z}$, thus generating sample set $\mathbf{p}\cap \mathbf{z}$. 
We have 
\[
(1-\gamma)\cdot m'^2\cdot s'\le \vert \mathbf{p}\cap \mathbf{z} \vert \le (1+\gamma)\cdot m^2\cdot s'
\]
with probability at least $1-2n^{-10}$ when sampling $s'\ge\frac{30\ln(n)}{\gamma^2m'^2}$ elements.
\end{lemma}

\begin{proof}
The result is an immediate consequence of the Chernoff bound: when we sample $s'\ge~ \frac{20\ln(n)}{\gamma^2m'^2}$ elements, we have with probability at least $1-n^{-10}$ that $(1-\gamma)\left(m'\cdot s\right)^2~\cdot~\frac{s'}{s^2}\le~ \vert\mathbf{p}~\cap~\mathbf{z}\vert$. A similar result gives the upper bound on $\vert \mathbf{p}\cap \mathbf{z}\vert$ for $s'\ge \frac{30\ln(n)}{\gamma^2m^2}$. As $m'\le m$, we maximize $s'$ by  $\frac{30\ln(n)}{\gamma^2m'^2}$ and thus ensure both bounds with probability at least $1-2n^{-10}$ using a union bound.
\end{proof}

We are generally going to use $\gamma$ as the same fixed parameter (to be
determined later) every time we invoke the sampling of Lemma \ref{lem:vectorsizeNEW}.
\medskip

We will use Lemma \ref{lem:vectorsizeNEW} to show that sampling after squaring will not distort the Jaccard similarity of a pair of vectors too much, and hence we get the benefits of squaring without the exploding vector dimensions. We start by bounding the resulting sizes for each of $\vert \mathbf{a}\vert,\vert \mathbf{b}\vert$ and $\vert \mathbf{a}\cap \mathbf{b}\vert$ for any choice of $\mathbf{a},\mathbf{b}\in \mathcal{A}\times \mathcal{B}$ from squaring and sampling $i$ times.

\begin{lemma}
\label{lem:squaringsampling}
Let $\mathbf{v}$ be a set from a universe of size $d$ or the intersection of such two sets. Let $\mathbf{f(v,i)}$ denote the resulting set after running $i$ iterations of the squaring-and-sampling reduction on set $\mathbf{v}$ for $i\ge 1$. We have
\[
(1-\gamma)^{2^i}\frac{\vert \mathbf{v}\vert^{2^i}}{d^{2^i}}s_i\le \vert \mathbf{f(v,i)}\vert\le (1+\gamma)^{2^i}\frac{\vert \mathbf{v}\vert^{2^i}}{d^{2^i}}s_i
\]
with probability at least $1-2in^{-10}$
where $s_i\ge \frac{30\ln(n)d^{2^i}}{\gamma^2(1-\gamma)^{2^i-2}\vert \mathbf{v}\vert^{2^i}}$.
\end{lemma}
\begin{proof}
Let $\mathbf{v}$ be as described. We show the lemma by induction on $i$. Clearly, when squaring the vector $\mathbf{v}$ once, i.e., for $i=1$, the resulting vector has Hamming weight $\vert \mathbf{v}\vert^2$ and dimension~$d^2$. Hence, by Lemma \ref{lem:vectorsizeNEW} we have
\[
(1-\gamma)\frac{\vert \mathbf{v}\vert^2}{d^2}\cdot s_1\le \vert \mathbf{f(v,1)}\vert\le (1+\gamma)\frac{\vert \mathbf{v}\vert^2}{d^2}\cdot s_1
\]
with probability at least $1-2n^{-10}$ for our choice of $s_1$. 
Assume now that after $i-1$ iterations the following bounds hold:
\begin{align}
\label{align:boundsonpriorstep}
(1-\gamma)^{2^{i-1}-1}\frac{\vert \mathbf{v}\vert^{2^{i-1}}}{d^{2^{i-1}}}s_{i-1}\le \vert \mathbf{f(v,i-1)}\vert\le (1+\gamma)^{2^{i-1}-1}\frac{\vert \mathbf{v}\vert^{2^{i-1}}}{d^{2^{i-1}}}s_{i-1}.
\end{align}
Then Lemma \ref{lem:vectorsizeNEW} gives that after $i$ iterations of the squaring-and-sampling reduction, we have 
\[
(1-\gamma)^{2^i-1}\frac{\vert \mathbf{v}\vert^{2^i}s_{i-1}^2}{d^{2^i}}\cdot \frac{s_i}{s_{i-1}^2}\le \vert \mathbf{f(v,i)}\vert\le (1+\gamma)^{2^i-1}\frac{\vert \mathbf{v}\vert^{2^i}s_{i-1}^2}{d^{2^i}}\cdot \frac{s_i}{s_{i-1}^2}
\]
with probability at least $1-2n^{-10}$ for $s_i\ge \frac{30\ln(n)d^{2^i}}{\gamma^2(1-\gamma)^{2^i-2}\vert \mathbf{v}\vert^{2^i}}$.
This particularly means that
\[
(1-\gamma)^{2^i}\frac{\vert \mathbf{v}\vert^{2^i}}{d^{2^i}}\cdot s_i\le \vert \mathbf{f(v,i)}\vert\le (1+\gamma)^{2^i}\frac{\vert \mathbf{v}\vert^{2^i}}{d^{2^i}}\cdot s_i.
\]
Now, to ensure these bounds, we assumed that $\vert \mathbf{f(v,i-1)}\vert$ satisfies certain bounds (see (\ref{align:boundsonpriorstep})). So in order to ensure that $\mathbf{f(v,i)}$ satisfies the given bounds, we need $\mathbf{f(v,j)}$ to satisfy similar bounds for every $1\le j\le i$. By a union bound, we see that $\vert \mathbf{f(v,j)}\vert$ satisfies both upper and lower bounds for all $1\le j\le i$ (simultaneously) with probability at least $1-2in^{-10}$ when sampling $s_j\ge \frac{30\ln(n)d^{2^j}}{\gamma^2(1-\gamma)^{2^j-2}\vert \mathbf{v}\vert^{2^j}}$ at step $j$.
Hence, $\vert \mathbf{f(v,i)}\vert$ satisfies the given bound with probability at least $1-2in^{-10}$.
\end{proof}

The next section will use Lemma \ref{lem:squaringsampling} to bound the Jaccard similarity after $i$ iterations of the squaring/sampling reduction.

\subsection{Combining the bounds}
\label{sec:jaccbounds}
For a given pair of vectors $\mathbf{a}$ and $\mathbf{b}$, Lemma \ref{lem:squaringsampling} gives upper and lower bounds on the Jaccard similarity $J = J\Big(\mathbf{f(a,i)},\mathbf{f(b,i)}\Big)$. We claim that with probability at least $1-6in^{-10}$:
\begin{align*}
     J &\ge\frac{(1-\gamma)^{2^i-1}\frac{\vert \mathbf{a}\cap \mathbf{b}\vert^{2^i}}{d^{2^i}}s_i}{(1+\gamma)^{2^i-1}\frac{\vert \mathbf{a}\vert^{2^i}}{d^{2^i}}s_i+(1+\gamma)^{2^i-1}\frac{\vert \mathbf{b}\vert^{2^i}}{d^{2^i}}s_i-(1-\gamma)^{2^i-1}\frac{\vert \mathbf{a}\cap \mathbf{b}\vert^{2^i}}{d^{2^i}}s_i} \\&\ge\frac{(1-\gamma)^{2^i}\vert \mathbf{a}\cap \mathbf{b}\vert^{2^i}}{(1+\gamma)^{2^i}\left(\vert \mathbf{a}\vert^{2^i}+\vert \mathbf{b}\vert^{2^i}\right)-(1-\gamma)^{2^i}\vert \mathbf{a}\cap \mathbf{b}\vert^{2^i}}
\end{align*}

\begin{align*}
    J &\le \frac{(1+\gamma)^{2^i-1}\frac{\vert \mathbf{a}\cap \mathbf{b}\vert^{2^i}}{d^{2^i}}s_i}{(1-\gamma)^{2^i-1}\frac{\vert \mathbf{a}\vert^{2^i}}{d^{2^i}}s_i+(1-\gamma)^{2^i-1}\frac{\vert \mathbf{b}\vert^{2^i}}{d^{2^i}}s_i-(1+\gamma)^{2^i-1}\frac{\vert \mathbf{a}\cap \mathbf{b}\vert^{2^i}}{d^{2^i}}s_i}\\ &\le\frac{(1+\gamma)^{2^i}\vert \mathbf{a}\cap \mathbf{b}\vert^{2^i}}{(1-\gamma)^{2^i}\left(\vert \mathbf{a}\vert^{2^i}+\vert \mathbf{b}\vert^{2^i}\right)-(1+\gamma)^{2^i}\vert \mathbf{a}\cap \mathbf{b}\vert^{2^i}}
\end{align*}

\medskip
\noindent
This is easily seen by taking a union bound over the probabilities that each of $\vert \mathbf{a}\vert$, $\vert \mathbf{b}\vert$ and $\vert \mathbf{a}\cap \mathbf{b}\vert$ violate either the upper or the lower bound.
Next, we claim that these bounds imply:

\begin{align*}
    J \ge \frac{(1-\gamma)^{2^i}\vert \mathbf{a}\cap \mathbf{b}\vert^{2^i}}{(1+4\gamma)^{2^i}\left(\vert \mathbf{a}\vert^{2^i}+\vert \mathbf{b}\vert^{2^i}-\vert \mathbf{a}\cap \mathbf{b}\vert^{2^i}\right)}\ge \frac{(1-\gamma)^{2^i}\vert \mathbf{a}\cap \mathbf{b}\vert^{2^i}}{(1+\gamma)^{2^i}\left(\vert \mathbf{a}\vert^{2^i}+\vert \mathbf{b}\vert^{2^i}\right)-(1-\gamma)^{2^i}\vert \mathbf{a}\cap \mathbf{b}\vert^{2^i}}
\end{align*}

\begin{align*}
   J \le \frac{(1+\gamma)^{2^i}\vert \mathbf{a}\cap \mathbf{b}\vert^{2^i}}{(1-\gamma)^{2^i}\left(\vert \mathbf{a}\vert^{2^i}+\vert \mathbf{b}\vert^{2^i}\right)-(1+\gamma)^{2^i}\vert \mathbf{a}\cap \mathbf{b}\vert^{2^i}}\le \frac{(1+\gamma)^{2^i}\vert \mathbf{a}\cap \mathbf{b}\vert^{2^i}}{(1-4\gamma)^{2^i}\left(\vert \mathbf{a}\vert^{2^i}+\vert \mathbf{b}\vert^{2^i}-\vert \mathbf{a}\cap \mathbf{b}\vert^{2^i}\right)} .
\end{align*}

\medskip
\noindent
\final{The argument can be found in Appendix A in the full version on ArXiv \cite[App.~A]{arxiveversion}.}
\arxiv{The argument can be found in Appendix \ref{app:jaccardbounds}.}
In particular, we have argued for the following lemma. We ignore the sample size for now and discuss it in Section~\ref{sec:sumup}. 

\begin{lemma}
\label{lem:jaccardboundsSampling}
Let $\mathcal{A}$ and $\mathcal{B}$ be an instance of Bichromatic Closest Pair with Jaccard similarity. After applying the Squaring and Sampling mapping, $f$, $i$ times as previously described to each set in $\mathcal{A}$ and $\mathcal{B}$, we have for all $n^2$ pairs $(\mathbf{a},\mathbf{b})\in \mathcal{A}\times \mathcal{B}$ in the instance that:

\begin{align*}
    \left(\frac{1-\gamma}{1+4\gamma}\right)^{2^i} \frac{\vert \mathbf{a}\cap \mathbf{b}\vert^{2^i}}{\vert \mathbf{a}\vert^{2^i}+\vert \mathbf{b}\vert^{2^i}-\vert \mathbf{a}\cap \mathbf{b}\vert^{2^i}}&\le J\Big(\mathbf{f(a,i)},\mathbf{f(b,i)}\Big) \le
     \left(\frac{1+\gamma}{1-4\gamma}\right)^{2^i} \frac{\vert \mathbf{a}\cap \mathbf{b}\vert^{2^i}}{\vert \mathbf{a}\vert^{2^i}+\vert \mathbf{b}\vert^{2^i}-\vert \mathbf{a}\cap \mathbf{b}\vert^{2^i}}
\end{align*}

\medskip
\noindent
with probability at least $1-6in^{-8}$
\end{lemma}
Hence, with high probability none of the Jaccard similarities diverge too much from (\ref{JaccardSquaredNoSampling}) due to sampling. This was exactly what we wanted, as this allows us to reduce the dimension by sampling.

\subsection{Summing up}
\label{sec:sumup}
Recall that in our setting we reduce from instances where the set sizes of all red and blue sets are fixed. We now describe thresholds such that solving the instances constructed by the reduction $f$ cannot be done in truly subquadratic time. 
\begin{lemma}
\label{lem:sumup}
Let $\mathcal{A}$ and $\mathcal{B}$ be two collections of $n$ sets from a universe of dimension $d$, where all sets in $\mathcal{A}$ have size $y$ and all sets in $\mathcal{B}$ have size $z$. Assume that $(\mathcal{A},\mathcal{B})$ is taken from a family of instances of Bichromatic Closest Pair with Jaccard similarity, which require time $\Omega(n^{2-\delta})$ for thresholds $t_1=\frac{x_1}{y+z-x_1}$ and $t_2=\frac{x_2}{y+z-x_2}$. The reduction which applies $f$ $i$ times to each set in $\mathbf{s}\in \mathcal{A}\cup \mathcal{B}$ for $i\ge 1$ constructs an instance of Bichromatic Closest Pair with Jaccard similarity, which requires time $\Omega(n^{2-\delta})$ time for thresholds 
\[
t_1'=\left(\frac{1-\gamma}{1+4\gamma}\right)^{2^i}\frac{x_1^{2^i}}{y^{2^i}+z^{2^i}-x_1^{2^i}},\qquad \text{and}\qquad
t_2'=\left(\frac{1+\gamma}{1-4\gamma}\right)^{2^i}\frac{x_2^{2^i}}{y^{2^i}+z^{2^i}-x_2^{2^i}}.
\]
whose solution provides a valid solution to the original instance with high probability when sampling $s_j>\frac{30\ln(n)d^{2^j}}{\gamma^2(1-\gamma)^{2^j-2}x_2^{2^j}}$ at each step $1\le j\le i$.
\end{lemma}
\begin{proof}
Lemma \ref{lem:jaccardboundsSampling} ensures that with high probability a solution to the constructed instance provides a valid solution to the original instance, since no pair of sets is likely to have Jaccard similarities that deviate beyond the chosen thresholds.

In Lemma \ref{lem:jaccardboundsSampling} we skipped the discussion of the sample size at each iteration -- we will argue for it now. From Lemma \ref{lem:squaringsampling}, it is easily seen that we maximize the needed sample size for all of $\vert \mathbf{a}\vert$, $\vert \mathbf{b}\vert$ or $\vert \mathbf{a}\cap \mathbf{b}\vert$ for any choice of $\mathbf{a}$ and $\mathbf{b}$ in iteration $i$ by
\[
s_i>\frac{30\ln(n)d^{2^i}}{\gamma^2(1-\gamma)^{2^i-2}\min_{(\mathbf{a},\mathbf{b})\in \mathcal{A}\times \mathcal{B}}\left\{\vert \mathbf{a}\cap \mathbf{b}\vert\right\}^{2^i}}.
\]
Hence, sampling $s_i$ elements from the universe will ensure that each of the upper and lower bounds for either $\vert \mathbf{a}\vert$, $\vert \mathbf{b}\vert$ or $\vert \mathbf{a}\cap \mathbf{b}\vert$ will fail with probability at most $n^{-10}$ in that iteration. 
As $\min_{(\mathbf{a},\mathbf{b})\in \mathcal{A}\times \mathcal{B}}\left\{\vert \mathbf{a}\cap \mathbf{b}\vert\right\}$ is unknown, we instead use $x_2$, which was the intersection size for a pair with Jaccard similarity $j_2$. Such a pair need not exist, but as the set sizes are fixed, $x_2$ can be easily computed.

We have left to argue that the pairs with intersection smaller than $x_2$ also satisfy the bounds in Lemma \ref{lem:jaccardboundsSampling} with high probability. The main observation is that they only need to satisfy the upper bound, as the resulting Jaccard similarities need only to stay below the lower threshold, $t_2'$ --- the Jaccard similarities can become arbitrarily small without affecting the result.

\medskip
By bounding the size of each term as we did in Lemma \ref{lem:squaringsampling} using the chosen $s_i$, we see that the error probabilities are still at most $n^{-10}$ for each of $\vert \mathbf{a}\vert$, $\vert \mathbf{b}\vert$ and $\vert \mathbf{a}\cap \mathbf{b}\vert$ for any choice of $(\mathbf{a},\mathbf{b})\in \mathcal{A}\times \mathcal{B}$.
\end{proof}

\section{Main Result}
\label{sec:result}
We are now ready to prove Theorem \ref{thm:mainthm}. We first give some intuition behind the proof and state a few lemmas to ease the proof.
For convenience we restate Theorem \ref{thm:mainthm}.
\begingroup
\def\thetheorem{\ref{thm:mainthm}}
\begin{theorem}
Assuming the Orthogonal Vectors Conjecture (OVC), the following holds: for any $\delta>0$, there exists an $\varepsilon>0$ such that for any given $j_2<j_1<1-\delta$ satisfying $j_1\le j_2^{1-\varepsilon}$, solving Bichromatic Closest Pair with Jaccard similarity for $n$ red and $n$ blue sets for sets from a universe of size
$\ln(n)/j_2^{O(\log(1/j_1))}$
for thresholds $j_1$ and~$j_2$ requires time  $\Omega(n^{2-\delta})$.
\end{theorem}
\addtocounter{theorem}{-1}
\endgroup

\subsection{Intuition}
The proof of Theorem \ref{thm:mainthm} reduces instances of Bichromatic Closest Pair as described in Section~\ref{sec:baseinstances} by composing three reductions, that together construct instances of Bichromatic Closest Pair with Jaccard similarity, which requires time $\Omega(n^{2-\delta})$ for the given thresholds $j_1$ and $j_2$ and some~$\varepsilon$.
A short description of each of the reductions can be found in Section~\ref{sec:reductions}. Below, we give three lemmas showing that these reductions preserve hardness.

The first lemma states that adding common elements to all sets in the instance will preserve hardness. This reduction increases the Jaccard similarity of all pairs of red and blue sets, and by choice of the number of added elements, we ensure that pairs of sets that initially had Jaccard similarity higher than the \emph{lower} threshold will get Jaccard similarity greater than $1-\delta$.
Hence, we get hardness for thresholds that are greater than $1-\delta$. From this point we can decrease the thresholds using two other reductions to achieve the given thresholds, that by assumption are less than $1-\delta$.

The second lemma states that the squaring-and-sampling reduction, discussed in detail in Section~\ref{sec:squaringdetails}, preserves hardness. The squaring-and-sampling reduction allows us to decrease the thresholds, so they come close to $j_1$ and $j_2$. Finally, the third lemma states that the reduction, which adds elements to only red sets will still preserve hardness. This reduction ensures that we can decrease the Jaccard similarity further. We will use it in such a way, that we effectively multiply the upper bound by a well-chosen $\alpha$ that ensures that the upper threshold is $j_1$ after this reduction. The proof ends by picking an $\varepsilon$, such that $j_2$ is strictly greater than the current lower threshold, and thus preserves hardness for the thresholds $j_1$ and $j_2$.
    
\subsection{Supporting Lemmas}
In the following, assume that $\mathcal{A}$ and $\mathcal{B}$ are collections of  $n$ red and $n$ blue sets from a universe~$U$, respectively.

\begin{lemma}
\label{lem:deltamapping}
Let $0<\delta\le 1$ be given and let $(\mathcal{A},\mathcal{B})$ be any instance of Bichromatic Closest Pair with Jaccard similarity as described in Lemma \ref{lem:rubinstances}.
Define $\ell:=\max_{\mathbf{q}\in \mathcal{A}\cup \mathcal{B}}\{\vert \mathbf{q}\vert\}\cdot(1/\delta-1)$ and $\mathbf{x}:=\{x_1,...,x_\ell\}$ such that $\mathbf{x}\cap (\mathcal{A}\cup \mathcal{B})=\emptyset$, and further define the
mapping $g:\mathcal{A}\cup \mathcal{B}\rightarrow \mathcal{A}'\cup \mathcal{B}'$ by $g(\mathbf{v})=\mathbf{v}\cup \mathbf{x}$ where $\mathcal{A}'=\mathcal{A}\cup \mathbf{x}$ and equivalently $\mathcal{B}'=\mathcal{B}\cup \mathbf{x}$.
The reduction that applies $g$ to every element of $\mathcal{A}$ and $\mathcal{B}$ generates an instance $(\mathcal{A}',\mathcal{B}')$ of Bichromatic Closest Pair with Jaccard similarity that requires time $\Omega(n^{2-\delta})$ for some thresholds $t_1',t_2'\ge 1-\delta$.
\end{lemma}

\begin{proof}
First, note that if $\mathbf{v}\in \mathcal{A}$, then $g(\mathbf{v})\in \mathcal{A}'$ and similarly if $\mathbf{v}\in \mathcal{B}$ then $g(\mathbf{v})\in \mathcal{B}'$.
We recall that instances of Bichromatic Closest Pair as described in Lemma \ref{lem:rubinstances} are constructed such that all red sets have the same size and all blue sets have the same size. We also have $\max_{\mathbf{q}\in \mathcal{A}\cup \mathcal{B}}\{\vert \mathbf{q}\vert\}=\vert \mathbf{a}\vert$, for any $\mathbf{a}\in \mathcal{A}$, since the sets in $\mathcal{A}$ were larger than the sets in $\mathcal{B}$.
It is easy to see that hardness is preserved under the reduction. 
\medskip

We finally argue that the resulting thresholds are larger than $1-\delta$:
Let $(\mathbf{a},\mathbf{b})$ be any pair from $\mathcal{A}\times \mathcal{B}$ which has Jaccard similarity at least $t_2$ and let $\mathbf{a}'=g(\mathbf{a})$ and $\mathbf{b}'=g(\mathbf{b})$.
We argue that any such pair satisfies $\vert \mathbf{a}\cap \mathbf{b}\vert\ge \frac{\vert \mathbf{b}\vert}{2}$:
Note that with these particular instances of Bichromatic Closest Pair and from the proof of Lemma \ref{lem:basecase}, we have
\[
J(\mathbf{a},\mathbf{b})=\frac{\vert \mathbf{a}\cap \mathbf{b}\vert}{\vert \mathbf{a}\cup \mathbf{b}\vert}=\frac{\vert \mathbf{a}\cap \mathbf{b}\vert}{Tm+m-\vert \mathbf{a}\cap \mathbf{b}\vert}\ge t_2=\frac{t_1}{2}=\frac{1/T}{2}.
\]
Since $\vert \mathbf{b}\vert=m\ge \vert \mathbf{a}\cap \mathbf{b}\vert$, this implies
\[
\vert \mathbf{a}\cap \mathbf{b}\vert\ge \frac{m}{2}+\frac{m}{2T}-\frac{\vert \mathbf{a}\cap \mathbf{b}\vert}{2T}\quad \Rightarrow\quad \vert \mathbf{a}\cap \mathbf{b}\vert \ge m/2=\vert \mathbf{b}\vert/2.
\]
We will consider the Jaccard similarity of $\mathbf{a}'$ and $\mathbf{b}'$:

\begin{align*}
J(\mathbf{a}',\mathbf{b}')&= \frac{\vert \mathbf{a}\cap \mathbf{b}\vert+\vert \mathbf{a}\vert(1/\delta-1)}{\left(\vert \mathbf{a}\vert+\vert \mathbf{a}\vert(1/\delta-1)\right) +\left(\vert \mathbf{b}\vert+\vert \mathbf{a}\vert(1/\delta-1)\right)-\left(\vert \mathbf{a}\cap \mathbf{b}\vert+\vert
\mathbf{a}\vert(1/\delta-1)\right)} \\&=\frac{\vert \mathbf{a}\cap \mathbf{b}\vert+\vert \mathbf{a}\vert(1/\delta-1)}{\vert \mathbf{a}\vert/\delta +\vert \mathbf{b}\vert-\vert \mathbf{a}\cap \mathbf{b}\vert} 
\end{align*}

\medskip
\noindent
By assumption $\vert \mathbf{a}\cap \mathbf{b}\vert\ge \frac{\vert \mathbf{b}\vert}{2}$, so:

\begin{align*}
     \frac{\vert \mathbf{a}\cap \mathbf{b}\vert+\vert \mathbf{a}\vert(1/\delta-1)}{\vert \mathbf{a}\vert/\delta +\vert \mathbf{b}\vert-\vert \mathbf{a}\cap \mathbf{b}\vert}\ge \frac{\vert \mathbf{b}\vert/2+\vert \mathbf{a}\vert(1/\delta-1)}{\vert \mathbf{a}\vert/\delta +\vert \mathbf{b}\vert/2} &\ge 1-\delta\\ \qquad\Leftrightarrow\qquad \frac{\vert \mathbf{b}\vert}{2}+\vert \mathbf{a}\vert(1/\delta-1)&\ge \vert \mathbf{a}\vert(1/\delta-1) +\frac{\vert \mathbf{b}\vert}{2}-\frac{\vert \mathbf{b}\vert\delta}{2}
\end{align*}
\noindent
which is always satisfied.
Hence, $J(\mathbf{a}',\mathbf{b}')\ge 1-\delta$ for any choice of $\delta>0$, and so, we construct an instance where every pair with Jaccard similarity higher than $t_2$ will have Jaccard similarity higher than $1-\delta$. Thus, there are thresholds that are greater than $1-\delta$, that make the constructed instance hard.
\end{proof}

\begin{lemma}
\label{lem:squaringmapping}
Let $0<\delta\le 1$ be given and consider any instance of Bichromatic Closest Pair with Jaccard similarity, $(\mathcal{A},\mathcal{B})$, from a family of instances which require time $\Omega(n^{2-\delta})$ for thresholds $t_1$ and $t_2$.
Using the reduction $f$ defined in Section~\ref{sec:squaringdetails} on each $\mathbf{v}\in \mathcal{A}\cup \mathcal{B}$ for $i$ iterations where $i\ge 1$, we construct a valid instance of Bichromatic Closest Pair with Jaccard similarity with high probability, which requires time $\Omega(n^{2-\delta})$ for thresholds that are decreasing functions of $i$.
\end{lemma}
\begin{proof}
The lemma follows immediately from Lemma \ref{lem:sumup}.
\end{proof}

\begin{lemma}
\label{lem:alphamapping}
Let $0<\delta\le 1$ be given and consider any instance of Bichromatic Closest Pair with Jaccard similarity, $(\mathcal{A},\mathcal{B})$, from a family of instances which require time $\Omega(n^{2-\delta})$ for thresholds $t_1$ and $t_2$.
 Define $\ell:=\max_{\mathbf{q}\in \mathcal{A}\cup \mathcal{B}}\{\vert \mathbf{q}\vert\}\cdot(1/\alpha-1)$ and $\mathbf{y}:=\{y_1,...,y_\ell\}$ such that $\mathbf{y}\cap (\mathcal{A}\cup \mathcal{B})=\emptyset$. Define 
mapping $h:\mathcal{A}\rightarrow \mathcal{A}'$ where $\mathcal{A}'=\mathcal{A}\cup Y$ by $h(\mathbf{a})=\mathbf{a}\cup \mathbf{y}$.
The reduction that applies $h$ to every element of $\mathcal{A}$ generates an instance $(\mathcal{A}',\mathcal{B})$ of Bichromatic Closest Pair with Jaccard similarity that requires time $\Omega(n^{2-\delta})$ for some thresholds $t_1',t_2'$.
\end{lemma}

\begin{proof}
Clearly, hardness is preserved under the reduction that simply adds new elements to all red sets. In particular this reduction decreases the thresholds by decreasing the similarity between red and blue pairs.
\end{proof}

\subsection{Proof outline for Theorem \ref{thm:mainthm}}
\label{sec:theproof}
\begin{proof}
 \arxiv{For simplicity and readability we leave out most of the calculations --- details can be found in Appendix \ref{app:computations}.}
 \final{For simplicity and readability we leave out most of the calculations --- details can be found in Appendix B in the full version on ArXiv \cite[App.~B]{arxiveversion}.}

Let $\delta>0$ be given and let $j_1,j_2$ be given such that $j_2<j_1<1-\delta$. 
Take any instance of Bichromatic Closest Pair with Jaccard similarity satisfying the properties described in Lemma \ref{lem:rubinstances}. Recall from this lemma that $T=O\left(\frac{1}{\varepsilon}\right)$. 

Apply the reductions from first Lemma \ref{lem:deltamapping} to achieve an instance, which requires time $\Omega(n^{2-\delta})$ for thresholds greater than $1-\delta$. 
We wish to reduce to an instance that is hard for smaller thresholds $j_1$ and $j_2$. The reduction from Lemma \ref{lem:squaringmapping} is used to decrease the thresholds, where we pick the largest $i$, such that the resulting upper threshold $t_1$ is no smaller than $j_1$, i.e., $t_1\ge j_1$. This reduction decreases the thresholds until the upper threshold is only slightly greater than $j_1$. Now, let $\alpha=\frac{j_1}{t_1}$ and apply the reduction from Lemma \ref{lem:alphamapping} to ensure that the resulting upper threshold is now equal to $j_1$.
This eventually gives an instance of Bichromatic Closest Pair with Jaccard similarity, which cannot be solved in time $O(n^{2-\delta})$ for thresholds 
\begin{align*}
t_1'&=\alpha\left(\frac{1-\gamma}{1+4\gamma}\right)^{2^i}\left(\frac{\delta}{T}+1-\delta\right)^{2^i}\\ t_2'&=\left(\frac{1+\gamma}{1-4\gamma}\right)^{2^i}\frac{\left(\frac{\delta}{2T}+1-\delta\right)^{2^i}}{\frac{1}{\alpha}+\left(\frac{\delta}{T}+1-\delta\right)^{2^i}-\left(\frac{\delta}{2T}+1-\delta\right)^{2^i}}
\end{align*}
where we observe that by construction $t_1'=\alpha\cdot t_1=j_1$. \arxiv{We refer to Appendix \ref{app:computations} for the calculations.}
\final{We refer to Appendix B in the full version on ArXiv for the calculations \cite[App.~B]{arxiveversion}.}
So we have constructed an instance which is hard for thresholds $j_1$ and $t_2'$. 
 
 \medskip

Set $t_2^*=\left(\frac{1+\gamma}{1-4\gamma}\right)^{2^i}\left(\frac{\delta}{2T}+1-\delta\right)^{2^i}$. Then $t_2'<\alpha t_2^*$ and so the hardness for $t_1'=j_1$ and $t_2'$ implies hardness for $t_1'=j_1$ and $\alpha t_2^*$. 
We show that there is an $\varepsilon$ that only depends on $\delta$ such
that $\alpha t_2^*<j_2$. Then the hardness for $t_1'=j_1$ and $\alpha t_2^*$ implies hardness for the given $j_1$ and $j_2$.
\medskip

Note that we have chosen $\alpha \ge t_1$, since otherwise $i$ could not be maximal.  
So we have:

\begin{align*}
    \frac{\log (j_1)}{\log \left(\alpha t_2^*\right)}=\frac{\log\left(\alpha t_1\right)}{\log\left(\alpha t_2^*\right)}&\le \frac{\log\left( t_1^2\right)}{\log\left(t_1\cdot  \left(\frac{1+\gamma}{1-4\gamma}\right)^{2^i}\cdot\left(\frac{\delta}{2T}+1-\delta\right)^{2^i}\right)}\\&=\frac{2^i\cdot\log\left(\left(\frac{1-\gamma}{1+4\gamma}\right)^2\cdot\left(\delta/T+1-\delta\right)^2\right)}{2^i\cdot\log\left(\left(\frac{1-\gamma}{1+4\gamma}\right)\cdot\left(\delta/T+1-\delta\right) \left(\frac{1+\gamma}{1-4\gamma}\right)\cdot\left(\frac{\delta}{2T}+1-\delta\right)\right)}.
\end{align*}
We need to show that this expression is bounded by $1-\varepsilon$ for some $\varepsilon$ that depends on $\delta$, but not on $j_1$ and $j_2$. Observe that the factors $2^i$ cancel out and we may pick $\gamma$ small enough that it can essentially be ignored. \arxiv{We show in Appendix \ref{app:computations}}\final{We show in Appendix B in the full version on ArXiv \cite[App.~B]{arxiveversion}} that we can use any $\gamma<\min\left\{\frac{1}{2^{i+1}},\frac{\delta}{20T}\right\}$.
Then for given $\delta$, there exists an $\varepsilon$ such that the expression is bounded by $1-\varepsilon$, since $T$ can be considered a constant for a fixed $\delta$. Recall that $T$ was defined in Lemma \ref{lem:rubinstances}.
By the assumption $j_1\le j_2^{1-\varepsilon}$ we then have $\alpha t_2^*<j_2$. Then the hardness of $t_1'$ and $\alpha t_2^*$
where $t_1'=j_1$ and $\alpha t_2^*<j_2$, implies the desired hardness for the given $j_1$ and $j_2$.

\medskip
We finally argue about the size of the universe of the instance constructed by the compositions of reductions described. In the following, $d$ is the size of the universe of the initial instance of Bichromatic Closest Pair with Jaccard instance. In the proof of Lemma~\ref{lem:sumup}, we argued that we could use $x_2$, which was the size of the intersection for a pair with Jaccard similarity $j_2$, in the sample size $s_i$, which means that
            \begin{align*}
            s_i&\ge\frac{30\ln(n)d^{2^i}}{\gamma^2(1-\gamma)^{2^i}x_2^{2^i}}=\frac{30\ln(n)d^{2^i}}{\gamma^2(1-\gamma)^{2^i}(j_2(\vert a\vert +\vert b\vert-x_2))^{2^i}}\\&=\frac{30\ln(n)}{\gamma^2(1-\gamma)^{2^i}j_2^{2^i}}\cdot\left(\frac{\delta+1}{1 +\frac{\delta}{2T}}\right)^{2^i}.
            \end{align*}
            Again, the calculations can be found in \arxiv{Appendix \ref{app:computations}.} \final{Appendix B in the full version on ArXiv \cite[App.~B]{arxiveversion}.}
            Hence, the sets constructed by the composition of reductions come from a universe whose size is bounded by
            \[
            \vert U\vert \le s_i+s_i(1/\alpha-1)=\frac{s_i}{\alpha}\le\frac{30\ln(n)}{\gamma^2j_2^{2^i}}\left(\frac{\delta+1}{\left(\frac{\delta}{T}+1-\delta\right)\left(\frac{\delta}{2T}+1\right)}\right)^{2^i}\left(\frac{1+4\gamma}{(1-\gamma)^2}\right)^{2^i}
            \]
            By Assumption $t_1'^2<j_1\le t_1'$, which implies that $2^i=O\left(\frac{\log j_1}{\log c}\right)=O\left(\log \frac{1}{j_1}\right)$ for constant $c<1$. Hence, we conclude that the size of the universe is $\ln(n) / j_2^{O(\log 1/j_1)}$.
This finishes the proof of Theorem \ref{thm:mainthm}.
\end{proof}

\section{Final Comments}
\label{sec:finalcomments}
On a final note, we remark that one can obtain a result similar to Theorem \ref{thm:mainthm} for Braun-Blanquet similarity. Recall that we define Braun-Blanquet similarity for a pair of sets $(\mathbf{a},\mathbf{b})\in \mathcal{A}\times \mathcal{B}$ as 
\[
BB(\mathbf{a},\mathbf{b})=\frac{\vert \mathbf{a}\cap \mathbf{b}\vert}{\max\big\{\vert \mathbf{a}\vert,\vert \mathbf{b}\vert\big\}}\in[0,1]
\]
In fact, the proof is slightly simpler than the one given in Section~\ref{sec:theproof} and the calculations are somewhat nicer. The proof ideas, i.e., the choice and order of reductions, are exactly the same and should be easy to carry out by following the structure of the proof of Theorem \ref{thm:mainthm}.

\medskip

The main open problem we leave is whether existing upper bounds are near-optimal when~$\varepsilon$ is an arbitrary constant between 0 and 1. Our techniques only work when $\varepsilon$ is sufficiently small.



\bibliography{lipics-v2019-sample-article}

\begin{thebibliography}{10}

\bibitem{broder1997resemblance}
Andrei~Z Broder.
\newblock On the resemblance and containment of documents.
\newblock In {\em Compression and complexity of sequences 1997. proceedings},
  pages 21--29. IEEE, 1997.

\bibitem{chen2018hardness}
Lijie Chen.
\newblock On the hardness of approximate and exact (bichromatic) maximum inner
  product.
\newblock In {\em 33rd Computational Complexity Conference, {CCC} 2018, June
  22-24, 2018, San Diego, CA, {USA}}, pages 14:1--14:45, 2018.
\newblock URL: \url{https://doi.org/10.4230/LIPIcs.CCC.2018.14}, \href
  {http://dx.doi.org/10.4230/LIPIcs.CCC.2018.14}
  {\path{doi:10.4230/LIPIcs.CCC.2018.14}}.

\bibitem{chen2018equivalence}
Lijie Chen and Ryan Williams.
\newblock An equivalence class for orthogonal vectors.
\newblock In {\em Proceedings of the Thirtieth Annual {ACM-SIAM} Symposium on
  Discrete Algorithms, {SODA} 2019, San Diego, California, USA, January 6-9,
  2019}, pages 21--40, 2019.
\newblock URL: \url{https://doi.org/10.1137/1.9781611975482.2}, \href
  {http://dx.doi.org/10.1137/1.9781611975482.2}
  {\path{doi:10.1137/1.9781611975482.2}}.

\bibitem{christiani2017set}
Tobias Christiani and Rasmus Pagh.
\newblock Set similarity search beyond minhash.
\newblock In {\em Proceedings of the 49th Annual {ACM} {SIGACT} Symposium on
  Theory of Computing, {STOC} 2017, Montreal, QC, Canada, June 19-23, 2017},
  pages 1094--1107, 2017.
\newblock URL: \url{https://doi.org/10.1145/3055399.3055443}, \href
  {http://dx.doi.org/10.1145/3055399.3055443}
  {\path{doi:10.1145/3055399.3055443}}.

\bibitem{goel2013discovering}
Ashish Goel, Aneesh Sharma, Dong Wang, and Zhijun Yin.
\newblock Discovering similar users on twitter.
\newblock In {\em 11th Workshop on Mining and Learning with Graphs}, 2013.

\bibitem{gupta2013wtf}
Pankaj Gupta, Ashish Goel, Jimmy~J. Lin, Aneesh Sharma, Dong Wang, and Reza
  Zadeh.
\newblock {WTF:} the who to follow service at twitter.
\newblock In {\em 22nd International World Wide Web Conference, {WWW} '13, Rio
  de Janeiro, Brazil, May 13-17, 2013}, pages 505--514, 2013.
\newblock URL: \url{https://doi.org/10.1145/2488388.2488433}, \href
  {http://dx.doi.org/10.1145/2488388.2488433}
  {\path{doi:10.1145/2488388.2488433}}.

\bibitem{indyk1998approximate}
Piotr Indyk and Rajeev Motwani.
\newblock Approximate nearest neighbors: Towards removing the curse of
  dimensionality.
\newblock In {\em Proceedings of the Thirtieth Annual {ACM} Symposium on the
  Theory of Computing, Dallas, Texas, USA, May 23-26, 1998}, pages 604--613,
  1998.
\newblock URL: \url{https://doi.org/10.1145/276698.276876}, \href
  {http://dx.doi.org/10.1145/276698.276876} {\path{doi:10.1145/276698.276876}}.

\bibitem{rubinstein2018hardness}
Aviad Rubinstein.
\newblock Hardness of approximate nearest neighbor search.
\newblock In {\em Proceedings of the 50th Annual {ACM} {SIGACT} Symposium on
  Theory of Computing, {STOC} 2018, Los Angeles, CA, USA, June 25-29, 2018},
  pages 1260--1268, 2018.
\newblock URL: \url{https://doi.org/10.1145/3188745.3188916}, \href
  {http://dx.doi.org/10.1145/3188745.3188916}
  {\path{doi:10.1145/3188745.3188916}}.

\bibitem{valiant2015finding}
Gregory Valiant.
\newblock Finding correlations in subquadratic time, with applications to
  learning parities and the closest pair problem.
\newblock {\em J. {ACM}}, 62(2):13:1--13:45, 2015.
\newblock URL: \url{https://doi.org/10.1145/2728167}, \href
  {http://dx.doi.org/10.1145/2728167} {\path{doi:10.1145/2728167}}.

\bibitem{Williams18}
Virginia~Vassilevska Williams.
\newblock Some open problems in fine-grained complexity.
\newblock {\em {SIGACT} News}, 49(4):29--35, 2018.
\newblock URL: \url{https://doi.org/10.1145/3300150.3300158}, \href
  {http://dx.doi.org/10.1145/3300150.3300158}
  {\path{doi:10.1145/3300150.3300158}}.

\end{thebibliography}

\arxiv{\appendix
\section{Bounds on Jaccard similarity after squaring-and-sampling reduction}
\label{app:jaccardbounds}
We show that the bounds given in Section~\ref{sec:jaccbounds} hold. We want to show the upper bound
\begin{align}
\label{ineq:jaccboundUpper}
\frac{(1+\gamma)^{2^i}\vert \mathbf{a}\cap \mathbf{b}\vert^{2^i}}{(1-\gamma)^{2^i}\left(\vert \mathbf{a}\vert^{2^i}+\vert \mathbf{b}\vert^{2^i}\right)-(1+\gamma)^{2^i}\vert \mathbf{a}\cap \mathbf{b}\vert^{2^i}}\le\frac{(1+\gamma)^{2^i}\vert \mathbf{a}\cap \mathbf{b}\vert^{2^i}}{(1-4\gamma)^{2^i}\left(\vert \mathbf{a}\vert^{2^i}+\vert \mathbf{b}\vert^{2^i}-\vert \mathbf{a}\cap \mathbf{b}\vert^{2^i}\right)} 
\end{align}
and similarly the lower bound
\begin{align}
\label{ineq:jaccboundLower}
\frac{(1-\gamma)^{2^i}\vert \mathbf{a}\cap \mathbf{b}\vert^{2^i}}{(1+\gamma)^{2^i}\left(\vert \mathbf{a}\vert^{2^i}+\vert \mathbf{b}\vert^{2^i}\right)-(1-\gamma)^{2^i}\vert \mathbf{a}\cap \mathbf{b}\vert^{2^i}}\ge\frac{(1-\gamma)^{2^i}\vert \mathbf{a}\cap \mathbf{b}\vert^{2^i}}{(1+4\gamma)^{2^i}\left(\vert \mathbf{a}\vert^{2^i}+\vert \mathbf{b}\vert^{2^i}-\vert \mathbf{a}\cap \mathbf{b}\vert^{2^i}\right)} .
\end{align}
Observe that (\ref{ineq:jaccboundUpper}) is true when
\begin{align*}
(1-\gamma)^{2^i}\vert \mathbf{a}\vert^{2^i}+(1-\gamma)^{2^i}\vert \mathbf{b}\vert^{2^i}-(1+\gamma)^{2^i}\vert \mathbf{a}\cap \mathbf{b}\vert^{2^i}&\ge (1-4\gamma)^{2^i}\left(\vert \mathbf{a}\vert^{2^i}+\vert \mathbf{b}\vert^{2^i}-\vert \mathbf{a}\cap \mathbf{b}\vert^{2^i}\right)\\
\left((1-\gamma)^{2^i}-(1-4\gamma)^{2^i}\right)\left(\vert \mathbf{a}\vert^{2^i}+\vert \mathbf{b}\vert^{2^i}\right)&\ge \left((1+\gamma)^{2^i}-(1-4\gamma)^{2^i}\right)\vert \mathbf{a}\cap \mathbf{b}\vert^{2^i}\\
\vert \mathbf{a}\vert^{2^i}+\vert \mathbf{b}\vert^{2^i}&\ge \frac{(1+\gamma)^{2^i}-(1-4\gamma)^{2^i}}{(1-\gamma)^{2^i}-(1-4\gamma)^{2^i}}\cdot \vert \mathbf{a}\cap \mathbf{b}\vert^{2^i}
\end{align*}
which in particular holds if (note that this is a loose bound)
\begin{align*}
\frac{(1+\gamma)^{2^i}-(1-4\gamma)^{2^i}}{(1-\gamma)^{2^i}-(1-4\gamma)^{2^i}}\le 2\qquad \Leftrightarrow\qquad
(1+\gamma)^{2^i}\le 2(1-\gamma)^{2^i}-(1-4\gamma)^{2^i}
\end{align*}
Similarly we see that (\ref{ineq:jaccboundLower}) is true when
\begin{align*}
(1+\gamma)^{2^i}\left(\vert \mathbf{a}\vert^{2^i}+\vert \mathbf{b}\vert^{2^i}\right)-(1-\gamma)^{2^i}\vert \mathbf{a}\cap \mathbf{b}\vert^{2^i}&\le (1+4\gamma)^{2^i}\left(\vert \mathbf{a}\vert^{2^i}+\vert \mathbf{b}\vert^{2^i}-\vert \mathbf{a}\cap \mathbf{b}\vert^{2^i}\right)\\
\left((1+4\gamma)^{2^i}-(1-\gamma)^{2^i}\right)\vert \mathbf{a}\cap \mathbf{b}\vert^{2^i}&\le \left((1+4\gamma)^{2^i}-(1+\gamma)^{2^i}\right)\left(\vert \mathbf{a}\vert^{2^i}+\vert \mathbf{b}\vert^{2^i}\right)\\
\frac{(1+4\gamma)^{2^i}-(1-\gamma)^{2^i}}{(1+4\gamma)^{2^i}-(1+\gamma)^{2^i}}\cdot\vert \mathbf{a}\cap \mathbf{b}\vert^{2^i}&\le \vert \mathbf{a}\vert^{2^i}+\vert \mathbf{b}\vert^{2^i}
\end{align*}
which in particular holds if
\begin{align*}
\frac{(1+4\gamma)^{2^i}-(1-\gamma)^{2^i}}{(1+4\gamma)^{2^i}-(1+\gamma)^{2^i}}&\le 2\qquad \Leftrightarrow\qquad
2(1+\gamma)^{2^i}\le (1+4\gamma)^{2^i}+(1-\gamma)^{2^i}
\end{align*}
So we want to ensure that for any choice of $i$ there exists a $\gamma$ which satisfies
\begin{align}
\label{ineq:top}
(1+\gamma)^{2^i}&\le 2(1-\gamma)^{2^i}-(1-4\gamma)^{2^i}\\
\label{ineq:bottom}
    2(1+\gamma)^{2^i}&\le (1+4\gamma)^{2^i}+(1-\gamma)^{2^i}
\end{align}
Let's begin with (\ref{ineq:top}): we see that we can bound the right-hand side as follows using Bernoulli's inequality
\begin{align*}
   (1+\gamma)^{2^i}&\le2(1-2^i\gamma)-(1-4\cdot 2^i\cdot\gamma) <  2(1-\gamma)^{2^i}-(1-4\gamma)^{2^i}
   \end{align*}
   and so (\ref{ineq:jaccboundUpper}) holds if we can show that 
   \begin{align*}
   (1+\gamma)^{2^i}&\le 2(1-2^i\gamma)-(1-4\cdot 2^i\cdot\gamma)=2-2\cdot 2^{i}\gamma-1+4\cdot 2^i\cdot\gamma=1+ 2^{i+1}\gamma.
\end{align*}

Now let's consider (\ref{ineq:bottom}): as before, we can bound the right-hand side by
\begin{align*}
   2(1+\gamma)^{2^i}&\le (1+4\cdot 2^i\cdot \gamma)+(1-2^i\cdot \gamma)\le  (1+4\gamma)^{2^i}+(1-\gamma)^{2^i}
\end{align*}
  and so (\ref{ineq:jaccboundLower}) holds if we can show that
\begin{align}
\label{ineq:pickinggamma}
   (1+\gamma)^{2^i}\le 1+\frac{3}{2}\cdot 2^i\cdot \gamma 
\end{align}
Note that both (\ref{ineq:jaccboundUpper}) and (\ref{ineq:jaccboundLower}) are satisfied if $(1+\gamma)^{2^i}\le 1+\frac{3}{2}\cdot 2^i\cdot \gamma$, which holds for $\gamma<\frac{1}{2^{i+1}}$ for every choice of $i$.

\section{Proof details for Theorem \ref{thm:mainthm}}
\label{app:computations}
We here give the calculations at each step of the proof of Theorem \ref{thm:mainthm}. Given $\delta$, we construct instances of Bichromatic Closest Pair with Jaccard similarity which require time $\Omega(n^{2-\delta})$ for certain thresholds - we will mainly give the thresholds after each reduction, as the ultimate goal is to show that there exists $\varepsilon$ such that the constructed instance is hard for any choice of $j_1$ and $j_2$ with $j_1\le j_2^{1-\varepsilon}$. Hence we construct instances which preserve hardness under varying thresholds and finally achieve an instance, where we can find $\varepsilon$, and thus that we can ensure hardness for any given thresholds $j_1$ and $j_2$ which satisfy $j_1\le j_2^{1-\varepsilon}$. Each of the thresholds are indexed by a number (1 for upper thresholds and 2 for lower threshold) and a letter (d the thresholds after the reduction adding common elements to all sets, s for the thresholds after the squaring reduction and a for the thresholds after adding elements to the red sets only).

\medskip
Let $\delta>0$ be given. Let $j_2<j_1<1-\delta$ be given.
\begin{enumerate}
    \item Let $j_{01}$ and $j_{02}$ be as in the proof of Lemma \ref{lem:basecase}. Then Bichromatic Closest Pair with Jaccard similarity cannot be solved in time $O(n^{2-\delta})$ for thresholds $j_{01}$ and $j_{02}$. 
    We will preserve this hardness under a series of reductions.
    \item Now add $\max_{\mathbf{x}\in \mathcal{A}\cup \mathcal{B}}\{\vert \mathbf{x}\vert\}(1/\delta-1)$ common values to all sets  in $\mathcal{A}$ and $\mathcal{B}$. Since all sets in $\mathcal{A}$ have size $Tm$ and all sets in $\mathcal{B}$ have size $m$, we get an instance, which requires time $\Omega(n^{2-\delta})$ for thresholds
        \begin{align*}
            j_{1d}&=\frac{m+Tm(1/\delta-1)}{Tm+Tm(1/\delta-1)+m+Tm(1/\delta-1)-\left(m+Tm(1/\delta-1)\right)}\\&=\frac{m+Tm/\delta-Tm}{Tm+Tm/\delta-Tm}=\frac{\delta(1+T/\delta-T)}{T}=\delta/T+1-\delta>1-\delta\\
            j_{2d}&=\frac{m/2+Tm(1/\delta-1)}{Tm+Tm(1/\delta-1)+m+Tm(1/\delta-1)-\left(m/2+Tm(1/\delta-1)\right)}\\&=\frac{m/2+Tm/\delta-Tm}{Tm/\delta+m+Tm/\delta-Tm-m/2-Tm/\delta+Tm}\\&=\frac{\frac{1}{2T}+1/\delta-1}{1/\delta+\frac{1}{2T}}=\frac{\frac{\delta}{2T}+1-\delta}{1+\frac{\delta}{2T}}>1-\delta
        \end{align*}
        Note that both thresholds are greater than $1-\delta$.
        Clearly, Bichromatic Closest Pair with Jaccard similarity still requires time $O(n^{2-\delta})$ for thresholds $j_{1d}$ and $j_{2d}$.
    \item We now use the squaring-and-sampling reduction on each set in the current instance to reduce the Jaccard similarity between all pairs of sets. We let the thresholds be as follows where $i$ is maximal such that $j_{1s}\ge j_1$. In order to satisfy (\ref{ineq:pickinggamma}) we require that $\gamma \le 1/2^{i+1}$.
    \begin{align*}
        j_{1s}&=\frac{\left(\frac{1-\gamma}{1+4\gamma}\right)^{2^i}\left(m+Tm(1/\delta-1)\right)^{2^i}}{\left(Tm+Tm(1/\delta-1)\right)^{2^i}+\left(m+Tm(1/\delta-1)\right)^{2^i}-\left(m+Tm(1/\delta-1)\right)^{2^i}}\\
        &=\left(\frac{1-\gamma}{1+4\gamma}\right)^{2^i}\frac{\left(m+Tm/\delta-Tm\right)^{2^i}}{\left(Tm/\delta\right)^{2^i}}=\left(\frac{1-\gamma}{1+4\gamma}\right)^{2^i}\left(\delta/T+1-\delta\right)^{2^i}\\
        j_{2s}&=\frac{\left(\frac{1+\gamma}{1-4\gamma}\right)^{2^i}\left(m/2+Tm(1/\delta-1)\right)^{2^i}}{\left(Tm+Tm(1/\delta-1)\right)^{2^i}+\left(m+Tm(1/\delta-1)\right)^{2^i}-\left(m/2+Tm(1/\delta-1)\right)^{2^i}}\\&=\left(\frac{1+\gamma}{1-4\gamma}\right)^{2^i}\frac{\left(1/2+T/\delta-T\right)^{2^i}}{\left(T/\delta\right)^{2^i}+\left(1+T/\delta-T\right)^{2^i}-\left(1/2+T/\delta-T\right)^{2^i}}\\&=\left(\frac{1+\gamma}{1-4\gamma}\right)^{2^i}\frac{\left(\frac{1}{2T}+1/\delta-1\right)^{2^i}}{\left(1/\delta\right)^{2^i}+\left(1/T+1/\delta-1\right)^{2^i}-\left(\frac{1}{2T}+1/\delta-1\right)^{2^i}}.
    \end{align*}
    Hence, using the squaring-and-sampling reduction, we constructed an instance of Bichromatic Closest Pair with Jaccard similarity, which requires time $\Omega(n^{2-\delta})$ for thresholds $j_{1s}$ and $j_{2s}$. Note that $j_{1s}$ is only slightly larger than $j_1$. The next step will decrease the upper threshold to become equal to $j_1$, which will allow us to argue about $\varepsilon$.
        \item The final reduction adds $\max_{\mathbf{v}\in \mathcal{A}'\cup \mathcal{B}'}\{\vert \mathbf{v}\vert\}(1/\alpha-1)$ to all red sets, where we let $\alpha=\frac{j_1}{j_{1s}}$, since then $j_1=j_{1a}=\alpha j_{1s}$. The number of elements that we add to the red sets ensures that we get an instance of Bichromatic Closest Pair, which is hard for upper threshold $j_{1a}=\alpha j_{1s}=j_1$ and some lower threshold, $j_{2a}$.
        \begin{align*}
        j_{1a}&
        =\frac{\left(\frac{1-\gamma}{1+4\gamma}\right)^{2^i}\left(m+Tm/\delta-Tm\right)^{2^i}}{\left(Tm/\delta\right)^{2^i}+(Tm/\delta)^{2^i}(1/\alpha-1)+\left(m+Tm/\delta-Tm\right)^{2^i}-\left(m+Tm/\delta-Tm\right)^{2^i}}
        \\&=\left(\frac{1-\gamma}{1+4\gamma}\right)^{2^i}\alpha\frac{\left(m+Tm/\delta-Tm\right)^{2^i}}{(Tm/\delta)^{2^i}}\\&=\alpha\left(\frac{1-\gamma}{1+4\gamma}\right)^{2^i}\left(\delta/T+1-\delta\right)^{2^i}=\alpha j_{1s}\\
        j_{2a}&
        =\frac{\left(\frac{1+\gamma}{1-4\gamma}\right)^{2^i}\left(m/2+Tm/\delta-Tm\right)^{2^i}}{\left(Tm/\delta\right)^{2^i}+(Tm/\delta)^{2^i}(1/\alpha-1)+\left(m+Tm/\delta-Tm\right)^{2^i}-\left(m/2+Tm/\delta-Tm\right)^{2^i}}
        \\&=\left(\frac{1+\gamma}{1-4\gamma}\right)^{2^i}\frac{\left(\frac{1}{2T}+1/\delta-1\right)^{2^i}}{\frac{(1/\delta)^{2^i}}{\alpha}+\left(\frac{1}{T}+1/\delta-1\right)^{2^i}-\left(\frac{1}{2T}+1/\delta-1\right)^{2^i}}.
        \end{align*}
       Observe further that
        \begin{align*}
        j_{2a}<\alpha\left(\frac{1+\gamma}{1-4\gamma}\right)^{2^i}\left(\frac{\delta}{2T}+1-\delta\right)^{2^i}= \alpha j_2^*
        \end{align*}
        So Bichromatic Closest Pair with thresholds $j_1$ and $\alpha j_2^*$ still requires time $\Omega(n^{2-\delta})$, for a well-chosen $\gamma$. A simple calculation shows that any gamma $\gamma<\frac{\delta}{20T}$ suffices. Hence, we choose any $\gamma<\min\left\{\frac{\delta}{20T},\frac{1}{2^{i+1}}\right\}$ in order to also satisfy (\ref{ineq:pickinggamma}).
        \item We will now find $\varepsilon$ such that $j_1>\alpha j_2^{*1-\varepsilon}$ and so, by the assumption that $j_1\le j_2^{1-\varepsilon}$, we conclude that $j_2>j_2^{*}$. Hence, the constructed instance of Bichromatic Closest Pair is also hard for thresholds $j_1$ and $j_2$.
        
        Note that $\alpha\ge j_{1s}$ --- otherwise $\alpha=\frac{j_1}{j_{1s}}< j_{1s}$ and so $j_1<j_{1s}^2$, which contradicts the assumption that $i$ was maximal. We observe that there exists an $\varepsilon>0$ such that 
        \begin{align*}
            \frac{\log j_1}{\log \left(\alpha j_{2}^*\right)}&=\frac{\log\left(\alpha j_{1s}\right)}{\log\left(\alpha\left(\frac{1+\gamma}{1-4\gamma}\right)^{2^i}\left(\frac{\delta}{2T}+1-\delta\right)^{2^i}\right)}\\
            &<\frac{\log\left( j_{1s}^2\right)}{\log\left(j_{1s}\left(\frac{1+\gamma}{1-4\gamma}\right)^{2^i}\left(\frac{\delta}{2T}+1-\delta\right)^{2^i}\right)}\\
            &=\frac{\log\left( \left(\left(\frac{1-\gamma}{1+4\gamma}\right)(\delta/T+1-\delta)\right)^2\right)}{\log\left(\left(\frac{1-\gamma}{1+4\gamma}\right)\left(\frac{\delta}{T}+1-\delta\right)\left(\frac{1+\gamma}{1-4\gamma}\right)\left(\frac{\delta}{2T}+1-\delta\right)\right)}\\
            &=\frac{2 \log\left(\left(\frac{1-\gamma}{1+4\gamma}\right)(\delta/T+1-\delta)\right)}{\log\left(\left(\frac{1-\gamma}{1+4\gamma}\right)\left(\frac{\delta}{T}+1-\delta\right)\left(\frac{1+\gamma}{1-4\gamma}\right)\left(\frac{\delta}{2T}+1-\delta\right)\right)}\\
            &=\frac{\log\left(\left(\frac{1-\gamma}{1+4\gamma}\right)(\delta/T+1-\delta)\right)^2}{ \log\left(\left(\frac{1-\gamma}{1+4\gamma}\right)\left(\frac{\delta}{T}+1-\delta\right)\left(\frac{1+\gamma}{1-4\gamma}\right)\left(\frac{\delta}{2T}+1-\delta\right)\right)} = 1-\varepsilon \enspace .
        \end{align*}
        For $\gamma<\min\left\{\frac{\delta}{20T},\frac{1}{2^{i+1}}\right\}$ we have that $\varepsilon = \Theta(1/T)$.

       Since $\frac{\log j_1}{\log \left(\alpha j_{2}^*\right)}<1-\varepsilon$ and $\frac{\log j_1}{\log j_2}\ge 1-\varepsilon$ by assumption, we have $\alpha j_{2}^*<j_2$. So the hardness of $t_1'$ and $\alpha t_2^*$ implies hardness for $j_1$ and $j_2$ as we wanted.
    \end{enumerate}

    We finally discuss the size of the universe. This size depends on the number of elements sampled in the squaring-and-sampling reduction. Recall that we use the squaring-and-sampling reduction on an instance of Bichromatic Closest Pair, where we have already added common elements to all sets in $\mathcal{A}$ and $\mathcal{B}$ and which is hard for thresholds $j_{1d}$ and $j_{2d}$.
    
    Let $x_2$ be the size of the intersection between a red and a blue set that have Jaccard similarity $j_{2d}$. Note that we can easily compute \begin{align*}
        x_2&=\frac{\vert \mathbf{a}\vert +\vert \mathbf{b}\vert - d_H(\mathbf{a},\mathbf{b})}{2}\\&=\frac{(Tm+Tm(1/\delta-1))+(m+Tm(1/\delta-1))-Tm}{2}\\&=\frac{m}{2}+Tm(1/\delta-1)
    \end{align*}
    We recall from Lemma \ref{lem:sumup} that the sample size should be
    \begin{align*}
            s_i&\ge\frac{30\ln(n)d^{2^i}}{\gamma^2(1-\gamma)^{2^i}x_2^{2^i}}=\frac{30\ln(n)d^{2^i}}{\gamma^2(1-\gamma)^{2^i}(j_2(\vert \mathbf{a}\vert +\vert \mathbf{b}\vert-x_2))^{2^i}}\\&=\frac{30\ln(n)}{\gamma^2(1-\gamma)^{2^i}j_2^{2^i}}\left(\frac{2Tm+Tm(1/\delta-1)}{(Tm+Tm(1/\delta-1)) +(m+Tm(1/\delta-1))-(m/2+Tm(1/\delta-1))}\right)^{2^i}\\&=\frac{30\ln(n)}{\gamma^2(1-\gamma)^{2^i}j_2^{2^i}}\cdot\left(\frac{\delta+1}{1 +\frac{\delta}{2T}}\right)^{2^i}.
            \end{align*}
            Recalling that $\alpha \ge t_{1s}$ we conclude that the universe has size
            \begin{align*}
                \vert U\vert &\le s_i+s_i(1/\alpha-1)=\frac{s_i}{\alpha}\le \frac{30\ln(n)}{\gamma^2(1-\gamma)^{2^i}j_2^{2^i}}\cdot\left(\frac{\delta+1}{1 +\frac{\delta}{2T}}\right)^{2^i}\cdot \frac{1}{\left(\frac{1-\gamma}{1+4\gamma}\right)^{2^i}\left(\frac{\delta}{T}+1-\delta\right)^{2^i}}\\&=\frac{30\ln(n)}{\gamma^2j_2^{2^i}}\left(\frac{\delta+1}{\left(\frac{\delta}{T}+1-\delta\right)\left(\frac{\delta}{2T}+1\right)}\right)^{2^i}\left(\frac{1+4\gamma}{(1-\gamma)^2}\right)^{2^i}.
            \end{align*}
    
    {\bf Remark.} Since $\varepsilon = \Theta(1/T)$ we get the same dependence of $\varepsilon$ on $\delta$ as that of Rubinstein, discussed in~\cite[Remark 1.4]{rubinstein2018hardness}. We note that $\varepsilon$ also depends on the unspecified function $c(\delta)$ in OVC, so we are not able to express it as a function of $\delta$, but the value is at least exponentially decreasing in $1/\delta$.
    
    }

\end{document}